%% file: OPARworkshop.tex
\documentclass[conference]{IEEEtran}
\IEEEoverridecommandlockouts

\usepackage[ruled]{algorithm2e}
\usepackage{color,soul}

\SetAlFnt{\small}
\SetAlCapFnt{\small}
\SetAlCapNameFnt{\small}
\SetAlCapHSkip{0pt}
\IncMargin{-\parindent}

\usepackage{algorithmic}
\usepackage{amsfonts}
\usepackage{amsmath}
\usepackage{amsthm}

\newtheorem{lemma}{Lemma}

\usepackage{nth}
\usepackage{graphicx}
\usepackage{ragged2e}
\usepackage[caption = false]{subfig}

\def\BibTeX{{\rm B\kern-.05em{\sc i\kern-.025em b}\kern-.08em
    T\kern-.1667em\lower.7ex\hbox{E}\kern-.125emX}}

\begin{document}

\title{OPAR: Optimized Predictive and Adaptive Routing for Cooperative UAV Networks}

\author{
	\IEEEauthorblockN{
	Mohammed Gharib\IEEEauthorrefmark{1}, Fatemeh Afghah\IEEEauthorrefmark{1}, Elizabeth Bentley \IEEEauthorrefmark{2}
	}
	\IEEEauthorblockA{\IEEEauthorrefmark{1}School of Informatics, Computing and Cyber Systems, Northern Arizona University, Flagstaff, AZ, USA\\ E-mail:\{mohammed.algharib, fatemeh.afghah\}@nau.edu}

	\IEEEauthorblockA{\IEEEauthorrefmark{2} Air Force Research Laboratory Rome, NY, USA, E-mail: elizabeth.bentley.3@us.af.mil}}


\maketitle

\begin{abstract}
\input{abstract} 
\end{abstract}


\input{introduction}

\input{proposedAlgorithm}

\input{performanceEvaluation}

\input{conclusion}

\nocite{*}
\bibliographystyle{IEEEtran}
\bibliography{References}


\end{document}

%% file: abstract.tex
Cooperative UAV networks are becoming increasingly popular in military and civilian applications. Alas, the typical ad-hoc routing protocols, which aim at finding the shortest path, lead to significant performance degradation because of the 3-dimension highly-dynamic nature of UAV networks and the uneven distribution of nodes across the network. This paper proposes OPAR, an optimized predictive and adaptive routing protocol, to face this challenging problem. We model the routing problem with linear programming (LP), where the goal is to maximize network performance, considering the path lifetime and path-length together. This model relies on a precise link lifetime prediction mechanism. We support the LP problem with a lightweight algorithm to find the optimized solution with a computation complexity of $O(|E|^2)$, where $|E|$ is the number of network links. We evaluate the OPAR performance and compare it with the well-known routing algorithms AODV, DSDV, and OLSR to cover a wide range of proactive and reactive protocols as well as distance vector and link-state techniques. We performed extensive simulations for different network densities and mobility patterns using the ns-3 simulator. Results show that OPAR prevents a high volume of routing traffic, increases the successful delivery by more than $30\%$, improves the throughput $25\%$ on average, and decreases the flow completion time by an average of $35\%$ \footnote{Distribution A. Approved for public release: Distribution unlimited MSC/PA-2020-0251; 88ABW-2020-3485 on 01 Dec 2020.}.

%% file: introduction.tex
\section{Introduction}
Cooperative unmanned aerial vehicle (UAV) networks, also known as flying ad hoc networks (FANETs), are easy to deploy, low cost, and flexible networks \cite{ WalidSurvey,  UAV-CPSsurvey}. 
In many applications, the network is considered for surveillance missions where the UAVs are tasked with sending high volume data (video, images, and other sensory data) to other UAVs or a ground station. Critical UAV missions such as patrolling, where the UAVs are tasked with monitoring a wide area or performing reconnaissance missions in a region, is an application instance. In such applications, the UAVs may not follow pre-defined paths, and their motion should be randomized to some extent \cite{Kuiper,Orfanus}. Such applications involve reliable delivery of information and cannot afford to lose information as a likely case in common routing protocols. Multi-hop communication has been recognized as a common communication technique to transfer UAVs' observed information to the base station noting the limited communication range and bandwidth of the UAVs, particularly in large-scale UAV networks or in areas where communication infrastructure is not available.  

The majority of currently deployed cooperative routing protocols in UAV networks are borrowed from mobile ad hoc networks (MANET) and vehicular ad hoc networks (VANET). They target finding the shortest path between the source and destination nodes.  Ad-hoc on-demand distance vector (AODV), destination sequenced distance vector (DSDV), and optimized link state routing (OLSR) are some instances. However, these routing algorithms are deemed inefficient in UAV networks as they are not customized to these networks' specific attributes. The main difference between the UAV networks and other typical multi-hop networks is their highly dynamic nature due to the higher speed and speed variability of nodes. The free movement in the three-dimension space and energy constraints are the other reasons which impact the link quality and cause link outages. Consequently, the network experiences multiple reroute processes, which impose a high volume of routing traffic overhead. 
This overhead causes failures in many communication flows, significantly degrades the network performance, dramatically decreases the throughput, and increases the flow completion time (FCT). Furthermore, due to the limited queue size of cooperative nodes, the TCP protocol has to deal with a large number of packets arriving out-of-order as being dropped and order a re-transmission for them. A large amount of such out-of-order packet delivery is caused by the packets belonging to the same flow transferred from different paths due to the multiple reroute processes. This fact results in a noticeable network performance degradation. Hence, in highly dynamic networks with high mobility, the routing protocols based on the shortest path are not always the ideal solution \cite{survey4}. 

To minimize the number of rerouting processes in UAV multi-hop communications, we propose to consider path lifetime. We define the \textit{path lifetime} as the minimum link lifetime for the links that form a path. Since the path with the longest lifetime can be much longer than the shortest path, it might not be the ideal solution too. Therefore, in this work, we aim to find a path that optimizes the network performance in terms of throughput, flow completion time, and flow success rate by taking both the path length and path lifetime into account. 

Several research efforts have recently focused on developing routing solutions for cooperative UAV networks. Some of them are designed to support mission-specific UAV networks, and hence their solution cannot be generally used \cite{fueling, packageDelivery,redefine}. Just a couple of the proposed routing algorithms have accounted for the link reliability \cite{MIT,polsr,ON}. Alas, none of the previously reported works optimizes their routes to achieve the highest performance in terms of throughput, FCT, and flow success rate. This fact motivates us to generalize the routing problem of highly dynamic cooperative UAV networks by proposing OPAR, a general routing algorithm that takes both the path lifetime and path length into consideration to offer a reliable information flow delivery in highly dynamic networks.

We first define a general Boolean LP problem that optimizes the path length and path lifetime altogether. The LP problem relies on the link lifetime prediction. Thus, we propose a precise link lifetime prediction algorithm that considers the movement direction in polar and azimuthal directions, the speed, and the acceleration of the nodes. The proposed algorithm is lightweight and needs only three consecutive locations of the UAVs. 
We further propose a polynomial-time algorithm that finds the optimization problem's solution in a low time complexity with the worst-case complexity of $O(|E|^2)$. 

We assume that the UAVs have a reliable low bandwidth communication channel with the ground station. The UAVs send their locations to the ground station via a control message. The ground station is responsible for gathering this data and making the decision about the routes. Since the OPAR algorithm is lightweight, the ground station's role can be easily distributed among the UAVs. However, since there might not be a path between some network nodes, finding the optimal path may not be feasible in some cases, in the absence of the ground station. To show the effect of the proposed analytical model, in this work, we assume such a ground station.   

To evaluate the performance of OPAR, we use the network simulator ns-3 \cite{ns3}. We compare OPAR with three well-known routing algorithms, AODV, OLSR, and DSDV, through extensive simulations. We choose them to cover different types of routing algorithms, including both categories of proactive and reactive protocols as well as distance vector and link-state methodologies. As the performance metrics, we calculate the delivery success rate, the routing traffic, network throughput, and FCT. This evaluation shows that OPAR prevents the massive volume of routing overhead caused by the typical cooperative routing algorithms. It decreases the FCT by an average of $25 \%$ and improves the throughput of the network by up to three times the throughput of other routing algorithms. More importantly, the successful flow completion rate of OPAR is, in the worst case, $20\%$ more than other algorithms.   

\textbf{Contributions:} The main contribution of this work is to optimize the route selection in cooperative UAV networks to maximize the network performance in terms of throughput, FCT, and flow success rate. More specifically, this work $i)$ analytically models the routing problem in UAV networks by jointly maximizing the path lifetime and minimizing the path length; $ii)$ precisely predicts the link lifetime of the entire network; $iii)$ finds a lightweight algorithm to solve the optimization problem in polynomial time complexity; and $iv)$ extensively evaluates the proposed algorithms' performance by comparing it with the state-of-the-art through ns-3 simulations.

\textbf{System model:} We model this network with a graph $G(V,E)$, where $V$ and $E$ represent the set of graph vertices and graph edges, respectively. Each vertex in graph $G$ represents a UAV node, where each edge represents the direct communication between the corresponding vertices. The edge $e_{(i,j)}\in E$ means that node $n_j$ is in the communication range of node $n_i$. Since the nodes' communication range may differ from one to another, the graph links are directed. The UAV nodes are mobile; thus, each graph link has a lifetime, limited by the node's transmission range and distance from the neighbor. For the model to be general, the nodes can move randomly in any direction, hover to complete their tasks, land to recharge their battery or change their payload or move according to a pre-defined path. Indeed, less movement makes the problem more tractable. However, we consider the problem in its most general form.

%% file: proposedAlgorithm.tex
\section{Proposed Algorithm}
\label{sec::proposedAlgorithm}

The goal is to find a path that maximizes the network's performance, knowing that the path length is not the only effective metric. Path lifetime is the other impacting factor that should be taken into account. Assuming that we have a precise link lifetime prediction matrix $\mathcal{T}$, where $\tau_{(i,j)} \in \mathcal{T}$ is the lifetime of the edge $e_{(i,j)}$. Let us consider the variable $x_{(i,j)}$ as a Boolean variable specifying whether the link $e_{(i,j)}$ participates in the optimized path or not. On the one hand, the objective function of minimizing $\sum_{}^{}x_{(i,j)}$, such that the chosen edges form a path from the source to the destination, guarantees the shortest path. On the other hand, maximizing $\tau$ where $\tau$ is the path lifetime guarantees the path with the longest lifetime. If we design a multi-objective optimization problem that considers both of the mentioned metrics together, the solution will be an area of feasible solutions. Based on the priority of the objectives, the solution could be any points in this area. However, we designed the following Boolean LP that considers both metrics of interest, simultaneously. In the proposed optimization problem, parameters $w_1$ and $w_2$ are the weights defining the trade-off between the longer lifetime and shorter path where $w_1+w_2=1$, $0 < w_1\leq 1$, and $0 \leq w_2\leq 1$. The higher value of these parameters leads to their higher impact on the optimization. The weights $w_1$ and $w_2$ are network setting parameters affected by many parameters, including but not limited to, the network area, network density, transmission range, and network load. 
In this optimization problem, variable $T$ represents the inverse of path lifetime,  Constraints (\ref{Const1}) and (\ref{Const2}) guarantee that the source node has one outgoing edge and no ingoing edge, Similarly, Constraints (\ref{Const3}) and (\ref{Const4}) guarantee that the destination node has one ingoing edge and no outgoing. Constraint (\ref{Const5}) guarantees that every vertex that has an ingoing edge has to have an outgoing edge, except for the source and the destination nodes. Constraint (\ref{Const6}) is the main constraint that maximizes path lifetime. In the rest of this section, we describe how to find the solution to the proposed Boolean LP problem and how to form the matrix $\mathcal{T}$. 

\begin{eqnarray}
 \min_{} & & \sum\limits_{\substack{(i,j)\in E\\i\in\{1,\cdots,n\}\\j\in\{1,\cdots,n\}}}  w_1 x_{(i,j) } + w_2 T 
 \label{TimeObj} \nonumber\\
   \mbox{Subject To:} & &  \sum\limits_{(s,i)\in E} x_{(s,i)} = 1 \label{Const1} \\
 & & \sum\limits_{(i,s)\in E} x_{(i,s)} = 0 \label{Const2} \\
& &\sum\limits_{(i,d)\in E} x_{(i,d)} = 1 \label{Const3} \\
& & \sum\limits_{(d,i)\in E} x_{(d,i)} = 0 \label{Const4} \\
& & \sum\limits_{\substack{(i,j)\in E\\ i\neq s }} x_{(i,j)} = \sum\limits_{\substack{(j,k)\in E \\ j\neq d}} x_{(j,k)} \label{Const5} \\
 && T \geq \frac{x_{(i,j)}}{\tau_{(i,j)}} \label{Const6}\\
 & & x_{(i,j)} \in \{0,1\} \label{NewConst7}\\
  & & 0 \leq T \leq1 \label{NewConst8}
\end{eqnarray}
\textbf{Optimization problem solution:} The defined optimization problem is a Boolean LP (BLP), which is well-known to be an NP-complete problem \cite{karp21}. However, it is a special BLP case, where the following proposed algorithmic method can find its solution, i.e. the optimized path. 

Assuming that we have a precise link lifetime prediction, we first sort the matrix of edge lifetimes in descending order. We then perform a breadth-first search (BFS) in the graph $G(V,E)$ to find the shortest path, without considering the link lifetimes. Let us refer to the shortest path as $p_0$. We then calculate the value of the objective function for the shortest path. Considering the link with the shortest link lifetime in path $p_0$ as $e_{(i',j')}$, we remove all the links with the lifetime lower or equal to the $\tau_{(i',j')}$. We then perform a BFS again, on the new graph and calculate the objective value for the new path. If the new objective value is less than the previous one, we replace the previous path with the new one and remove the edges with the lowest lifetime. We repeat this procedure until the BFS algorithm fails in finding a new path.
We show in Lemma (\ref{lem::optimizedPath}) that the output of this algorithm is the path with the optimized objective function. We further show that the computational complexity of this algorithm is on the order $O(|E||V|+|E|^2)$, in Lemma (\ref{lem::complexity}).
\begin{lemma}
\label{lem::optimizedPath}
Assume that there is at least one path from the source node to the destination. The output of the proposed Algorithm is the optimal path, i.e. the solution of the proposed Boolean LP problem. 
\end{lemma}
\begin{proof}
We use proof by induction to show the correctness of this lemma. The output of the first iteration of the algorithm is the optimal path until then. At the first iteration, the algorithm will return the shortest path that, in comparison with the previous outputs, it is the optimal one. Now, we have to show that if the algorithm stops at the $(n+1)^{th}$ iteration and the output path of the $n^{th}$ iteration was the optimal path until then, the $(n+1)^{th}$ iteration will return the optimal path. We use proof by contradiction to show it. Assume that the $(n+1)^{th}$ iteration will return a non-optimal path, a path with a lower objective value in comparison with the $n^{th}$ iteration, but still non-optimal. It means that there is a path with a shorter path length or a longer path lifetime, which was discarded when we removed the links with a lower lifetime in comparison with the shortest path in the $n^{th}$ iteration. On the one hand, since the BFS always returns the shortest possible path, there was no discarded path with a shorter path length.
On the other hand, since we just remove the links with shorter link lifetimes, the link with a longer lifetime could not be discarded. Hence, we meet a contradiction, and the lemma is proved. 
\end{proof}
\begin{lemma}
\label{lem::complexity}
The computational complexity of the proposed algorithm is on the order $O(|E||V|+|E|^2)$.
\end{lemma}
\begin{proof}
We know that the worst-case complexity of the BFS algorithm is $O(|V|+|E|)$.  
The proposed algorithm of this section, in the worst case, needs to do BFS $|E|$ times. Hence, the worst-case complexity of the proposed algorithm is $O(|E||V|+|E|^2)$. 
\end{proof}
It is worthy to mention that since in UAV networks usually $|E|>|V|$, the complexity of this algorithm could be considered as $O(|E|^2)$. Furthermore, although we proved that the worst-case complexity is $O(|E||V|+|E|^2)$, in practice the average-case complexity is much lower than the worst case, due to removing a large number of edges at the first iteration. 

\textbf{Link lifetime prediction:} While OPAR can generally work with any link lifetime prediction algorithm, we utilize three-dimensional geometry to predict the link lifetime by knowing the UAV nodes' positions. 
We then form a $n\times n$ matrix $\mathcal{T}$ which contains the LLT prediction for each link $e_{(i,j)}$, i.e. $\tau_{(i,j)}$. 
Considering a three-dimensional sphere, the position of node $i$ at time $t_0$ is defined by the tuple $p_i(t_0)=(x_i,y_i,z_i)_{t_0}$. Accordingly, the direction of each node in a three-dimensional sphere is defined by two angles azimuthal $\alpha$ and polar $\theta$. Having two consecutive positions of the same UAV node, $p_i(t_1)$ and $p_i(t_2)$, these angles could be calculated by Equation (\ref{eq::angels}). To calculate the velocity of node $i$ having its position in two consecutive times $t_1$ and $t_2$, Equation (\ref{eq::3dSpeed}) has to be used.
\begin{equation}
\label{eq::angels}
\begin{small}
\begin{cases}
\alpha_i&=tan^{-1}({\frac{y_i(t_2)-y_i(t_1)}{x_i(t_2)-x_i(t_1)}})\\
\theta_i&=cos^{-1}({\frac{z_i(t_2)-z_i(t_1)}{\sqrt{[x_i(t_2)-x_i(t_1)]^2 + [y_i(t_2)-y_i(t_1)]^2 + [z_i(t_2)-z_i(t_1)]^2}}})\\
&=tan^{-1}({\frac{\sqrt{[x_i(t_2)-x_i(t_1)]^2 + [y_i(t_2)-y_i(t_1)]^2 }}{z_i(t_2)-z_i(t_1)}})
\end{cases}
\end{small}
\end{equation}
\begin{align}
\label{eq::3dSpeed}
&v_i= \\\nonumber
&\frac{\sqrt{[x_i(t_2)-x_i(t_1)]^2 + [y_i(t_2)-y_i(t_1)]^2 + [z_i(t_2)-z_i(t_1)]^2}}{t_2-t_1}
\end{align}

To consider the acceleration of the node, we need three consecutive positions of the nodes, to calculate node speed in two consecutive times, i.e. $v_{i}(t_1)$ and $v_{i}(t_2)$, using Equation (\ref{eq::3dSpeed}). Thereby, the acceleration could be easily calculated using equation (\ref{eq::acc}). Accordingly, to calculate the new node position after time $\Delta t$, equation (\ref{eq::3dNextPosAcc}) could be used.
\begin{equation}
\label{eq::acc}
a_i=\frac{v_i(t_2)-v_i(t_1)}{t_2-t_0}
\end{equation}    
\begin{align}
\label{eq::3dNextPosAcc}
&p_i(t_2+\Delta t)=\\\nonumber
&\begin{cases}
x_i(t_2+\Delta t)=x_i(t_2)+( v_i(t_2) \Delta t+\frac{1}{2}a_i \Delta t^2 ) sin(\theta_i) cos(\alpha_i)\\
y_i(t_2+\Delta t)=y_i(t_2)+(  v_i(t_2)\Delta t+\frac{1}{2}a_i \Delta t^2 ) sin(\theta_i) sin(\alpha_i)\\
z_i(t_2+\Delta t)=z_i(t_2)+( v_i(t_2)\Delta t+\frac{1}{2}a_i \Delta t^2 ) cos(\theta_i) 
\end{cases}
\end{align} 

The euclidean distance between UAV $i$ and UAV $j$ at time $t_2+\Delta t$ could be calculated using Equation (\ref{eq::distance}). 
Let us assume, without loss of generality, the transmission range of all UAV nodes is the constant $R$. For each pair of nodes, if the distance is greater than $R$, we consider the LLT equal to zero. Now, our goal is to find the value of the maximum $(\tau_{(i,j)}=\Delta t)$ that keeps the $d_{(i,j)}(t_2+\Delta t)$ less than $R$ for all the remaining links. Hence, we have to find the maximum positive root of Equation (\ref{eq::root}).  The positions of nodes $i$ and $j$ in this equation for the time $(t_2+\Delta t)$ could be calculated using Equation (\ref{eq::3dNextPosAcc}). In Equation (\ref{eq::root}), all the values except $\Delta t$ are known. This equation has two roots, a positive and a negative root. Any low complexity numerical methods such as the Bisection method could be used to find the positive root. 
\begin{equation}
\label{eq::distance}
d_{(i,j)}(t_2+\Delta t)=\\
\sqrt{ \begin{aligned} 
&(x_i(t_2+\Delta t)-x_j(t_2+\Delta t))^2+\\
&(y_i(t_2+\Delta t)-y_j(t_2+\Delta t))^2+\\
&(z_i(t_2+\Delta t)-z_j(t_2+\Delta t))^2
\end{aligned}
}
\end{equation}
\begin{equation}
\label{eq::root}
d_{(i,j)}(t_2+\Delta t)-R=0
\end{equation}

%% file: performanceEvaluation.tex
\section{Performance Evaluation}
\label{sec::evaluation}

To show the superiority of OPAR, we evaluate its performance and compare it with three well-known routing protocols, AODV, OLSR, and DSDV. We use a wide range of routing protocols, including both distance vector and link-state techniques, as well as reactive and proactive protocols to show the superiority of OPAR in comparison with different routing protocols. We use network simulator ns-3 to implement OPAR and make the comparison. We measure flow success rate, routing traffic overhead, network throughput, and flow completion time (FCT) as performance evaluation metrics. The UAV nodes move randomly according to the 3D random waypoint (RWP) model, as well as the 3D Gauss-Markov (G-M) model. While RWP generates completely random movements in the area, the Gauss-Markov is a memory-based model and prevents the UAV from significantly changing its angle of movement. Table (\ref{tbl::simulationSetting}) represents the summary of simulation setting.
\begin{table}[t]
\caption[]{Simulation Setting }
\resizebox{1\textwidth}{!}{
\begin{minipage}{\textwidth}
\begin{tabular}{ l | l  }
  Number of UAVs & $[50\quad 100]$ \\
  Area size & $300\times 2000\times 50 m$\\
  Transmission power & 7.5 dBm\\
  Number of concurrent flows & $[1 \quad 10]$\\
  File size & 5 MB\\
  Simulation time & 500 sec\\
  Speed range & $[0 \quad 50] m/s$\\ 
  Mobility models & 3D RWP and 3D G-M\\
  Azimuthal range in G-M model & $[0 \quad 0.05]$ rad\\
  Routing protocols& OPAR, AODV, DSDV, OLSR\\
  Traffic type & TCP NewReno\\
  Wireless communication standard & IEEE 802.11b\\
  Propagation loss model & Free-space propagation loss\\ 
  Propagation delay model & Constant speed propagation delay
\end{tabular}
\label{tbl::simulationSetting}
\end{minipage}}
\end{table}

In each communication flow, a file of $5$ $MB$ is sent from the source UAV to the destination. We simulate each instance for 500 seconds and consider the flow status as a failure if the network fails at delivering the total file at the given simulation time. Each simulation instance has been simulated 10 times, with different randomized seed, and the average values of the results are reported. In the OPAR algorithm, each UAV records its own location every 0.3 sec and sends the three consequent locations via a message to the ground station. 

We first analyze the OPAR performance for different network loads and densities by varying the path length-weight, i.e. $w_1$. The aim is to find the effect of path length versus the effect of path lifetime in different loads and densities. Fig. (\ref{fig::wthroughput}) shows the network throughput for a different number of flows and different number of UAVs. For better representation, we show the results only for a network with one, five, and ten concurrent flows and for the network with 50, 75, and 100 UAV nodes. Network throughput is calculated as the ratio of successful delivery over the network bandwidth over time and measured in Mbps. It is worth mentioning that if there does not exist a path between the source and the destination nodes, the source node stops sending. As an instance, in a network with ten concurrent flows, if two of the destination nodes are physically unreachable, the corresponding source nodes stop sending their packets which leads to mistakenly reporting the throughput for the network with eight flows, instead of the network with ten concurrent flows. To report fair throughput results, the throughput has to be considered in accordance with the flow failure rate. Hereupon, we weighted the throughput with the corresponding failure rate. 

As Fig. (\ref{fig::wthroughput}a) shows, in the simulated settings, $w_1=0.6$ for the network with the lowest load leads to the highest throughput. This value of $w_1$ for the network with the highest load is between 0.8 and 0.9. Fig. (\ref{fig::wthroughput}b) shows that in the dense network, the weight close to one shows the highest throughput where this value for the network with medium density and low density is around 0.5 and 0.7, respectively. As it is clear, the optimal point between the path length weight and the path lifetime differs based on the network setting. However, the shortest path, in most cases, is not the ideal solution. 

\begin{figure}[t!]
	\centering
	\subfloat[Different number of flows.]{\includegraphics[width=.5\linewidth]{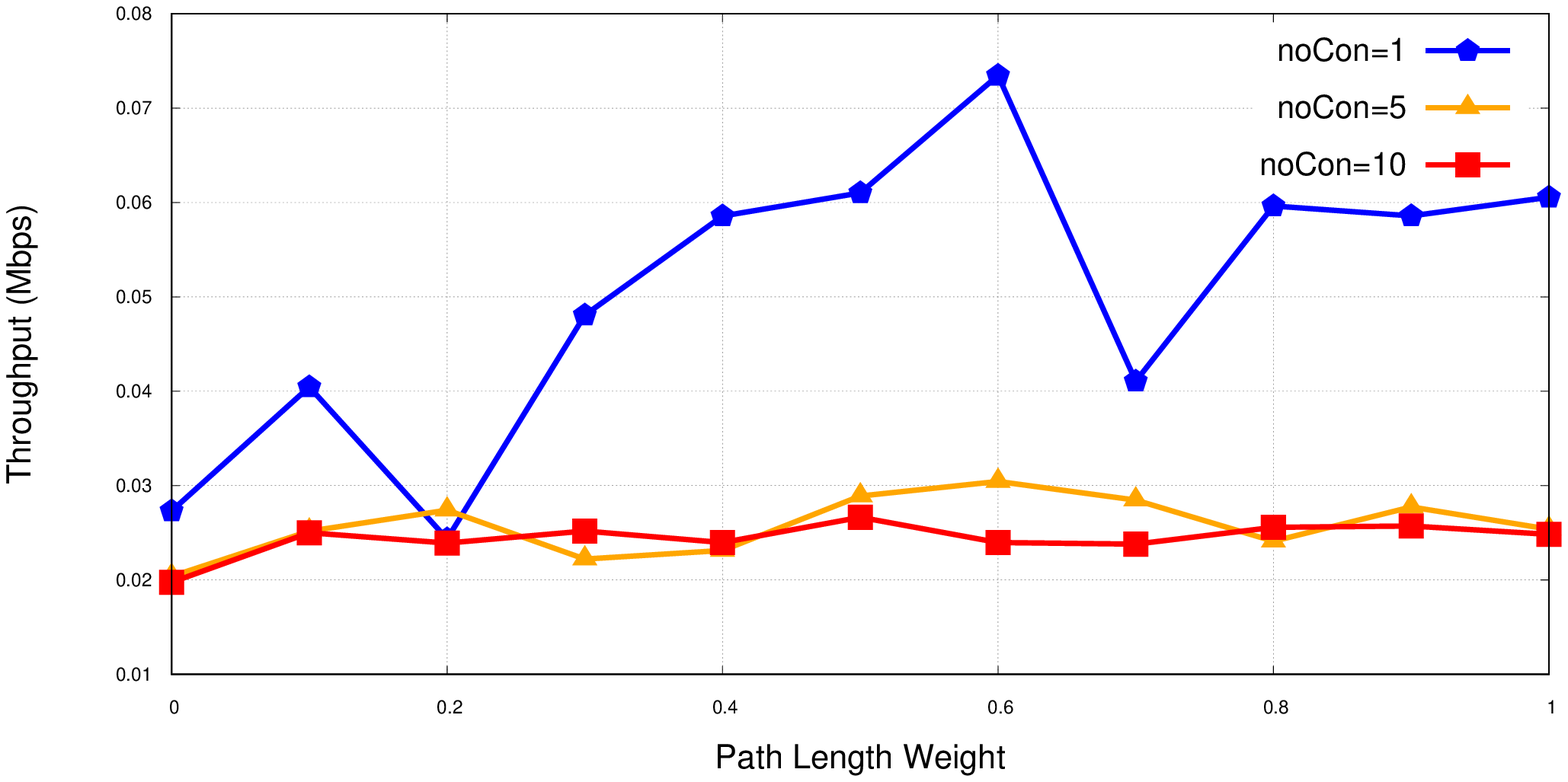}}
	\subfloat[Different number of UAVs.]{  \includegraphics[width=.5\linewidth]{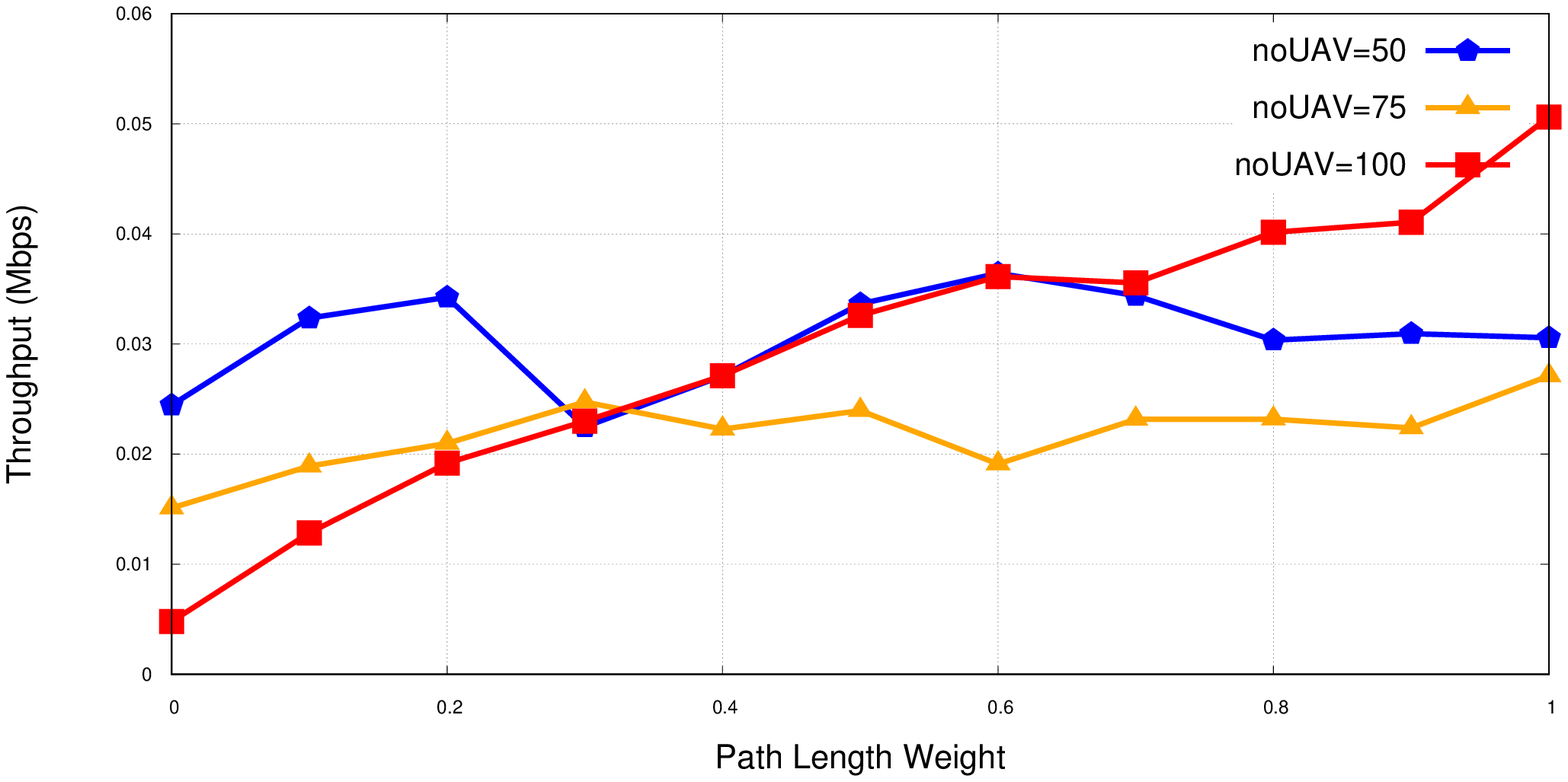}}
	\caption{A comparison of the throughout for different loads and densities.	
	}
	\label{fig::wthroughput}
\end{figure}
\begin{figure}[t!]
	\centering
	\subfloat[Different number of flows.]{\includegraphics[width=.5\linewidth]{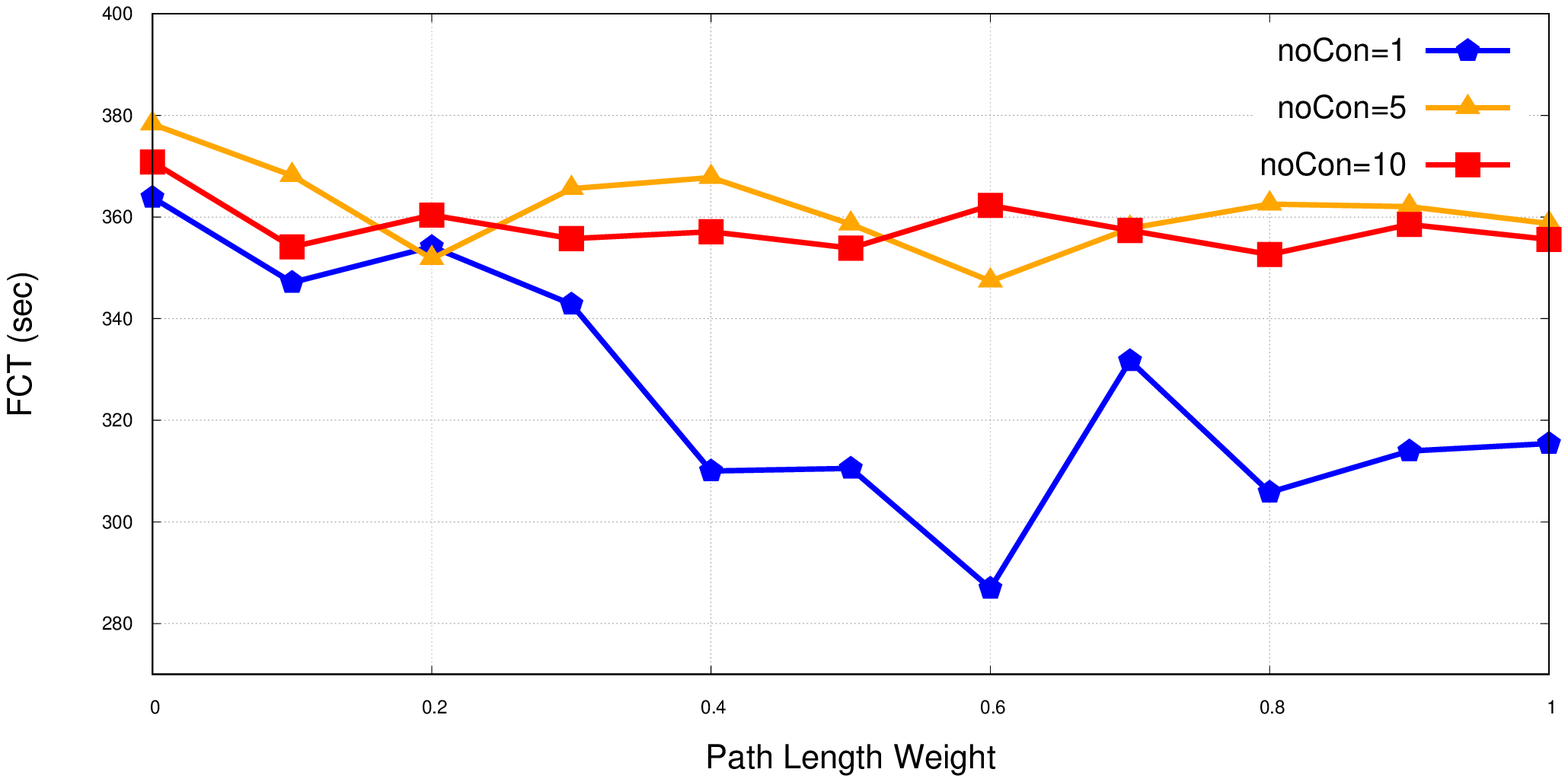}}
	\subfloat[Different number of UAVs.]{  \includegraphics[width=.5\linewidth]{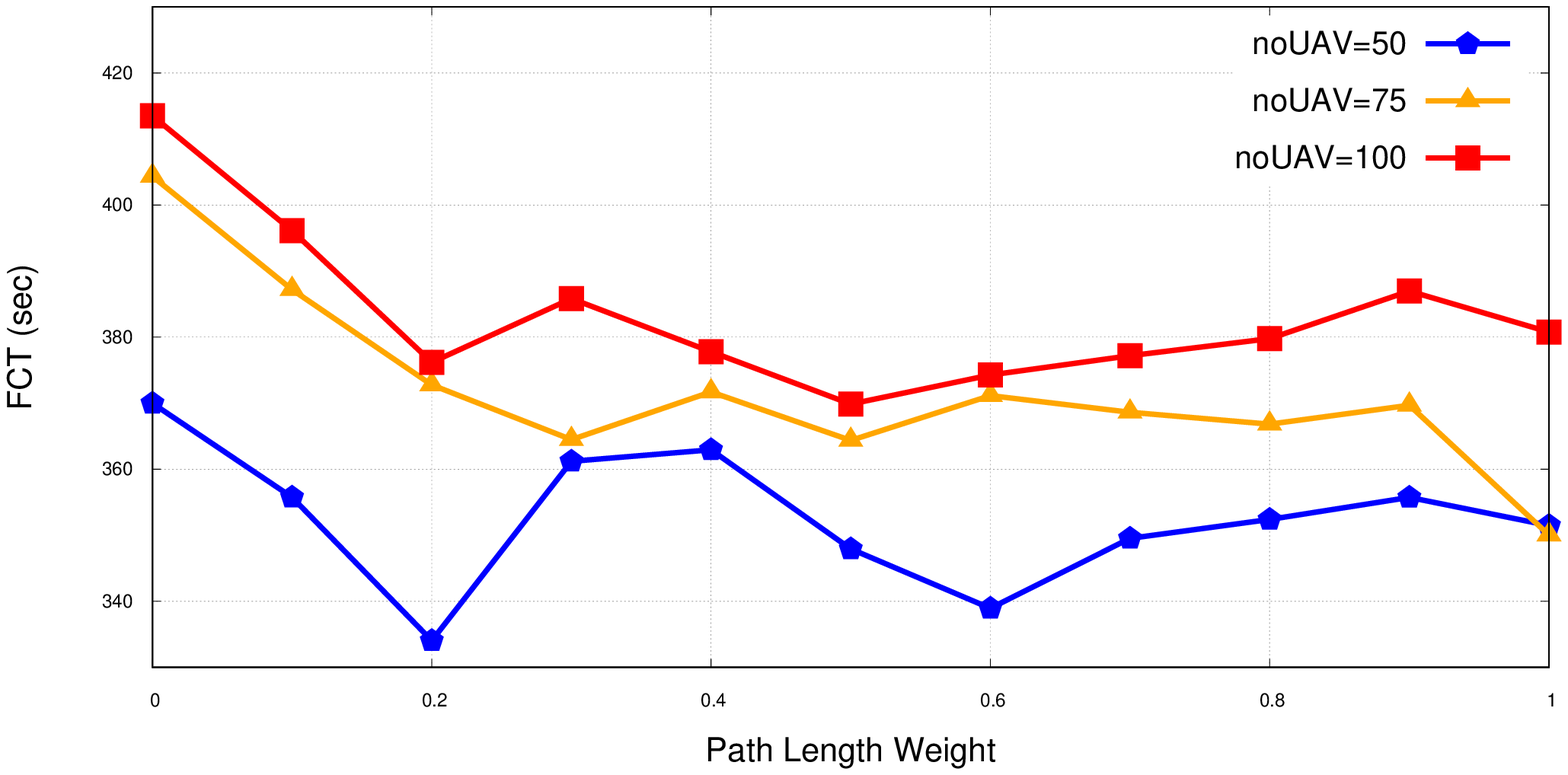}}
	\caption{A comparison of FCT for different loads and densities.}
	\label{fig::wFCT}
\end{figure}
\begin{figure}[ht!]
	\centering
	\subfloat[Different number of  flows (RWP)]{\includegraphics[width=.5\linewidth]{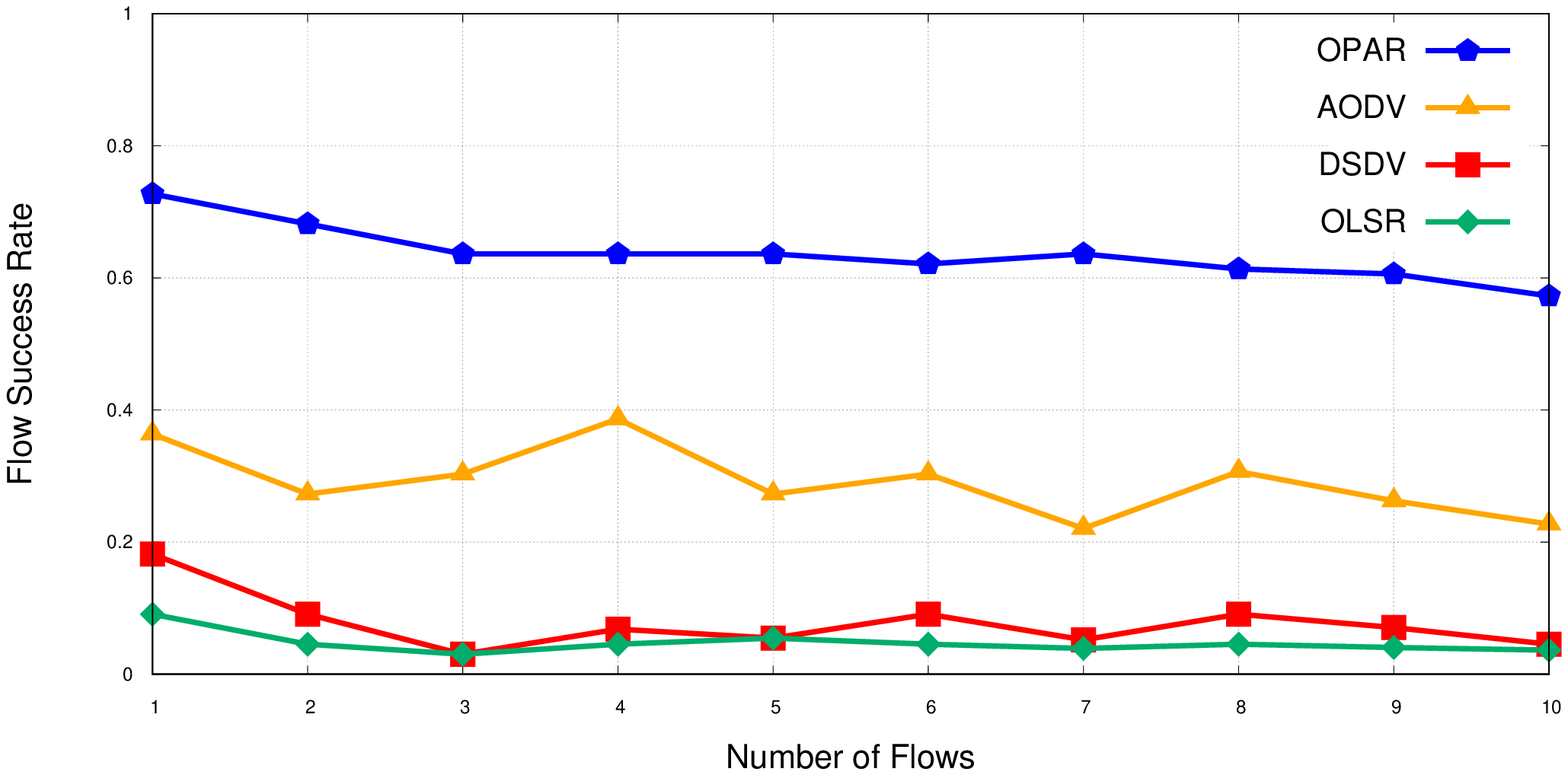}}
	\subfloat[Different number of  flows (G-M)]{\includegraphics[width=.5\linewidth]{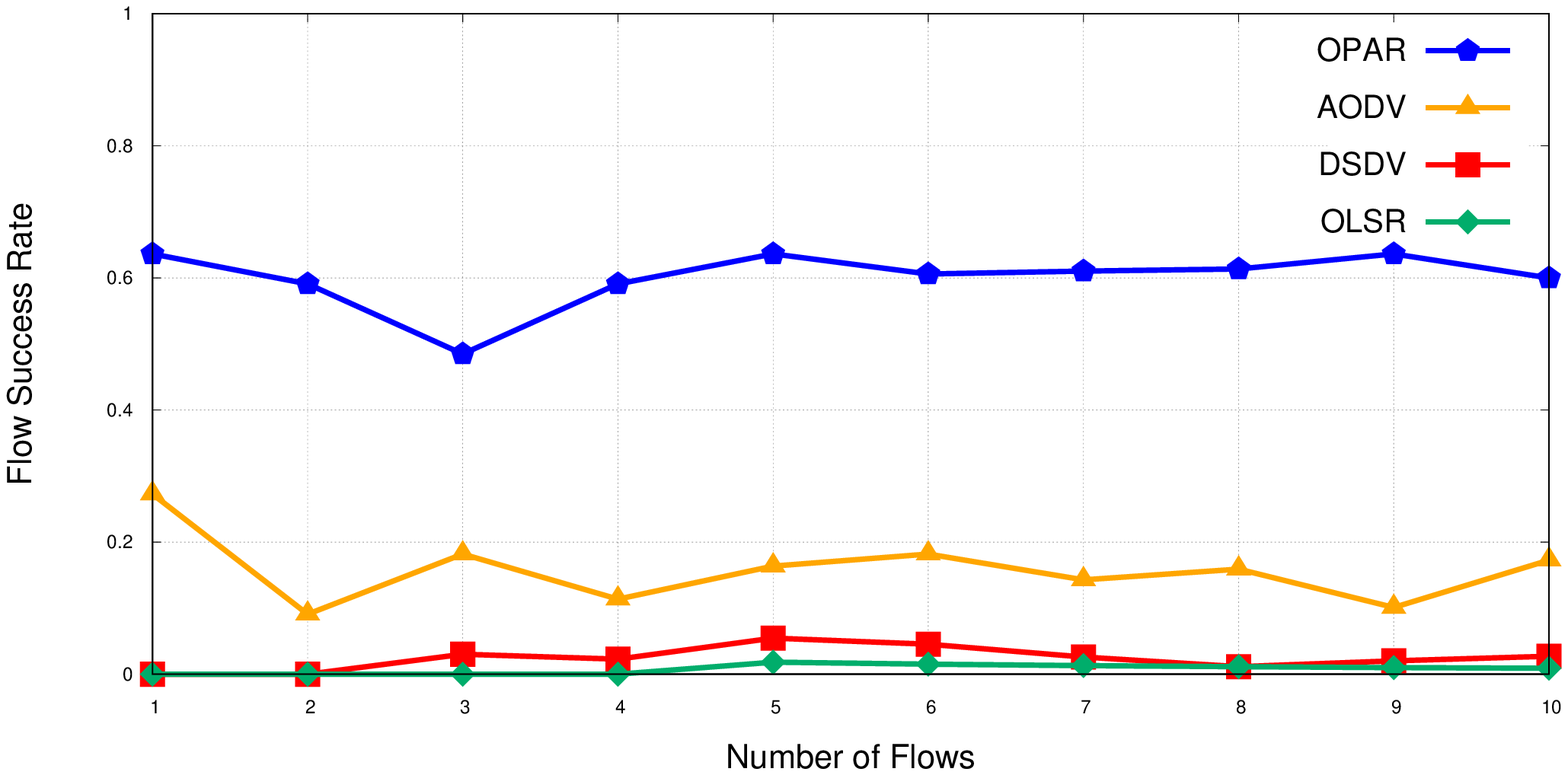}}\\
	\subfloat[Different number of UAVs (RWP)]{  \includegraphics[width=.5\linewidth]{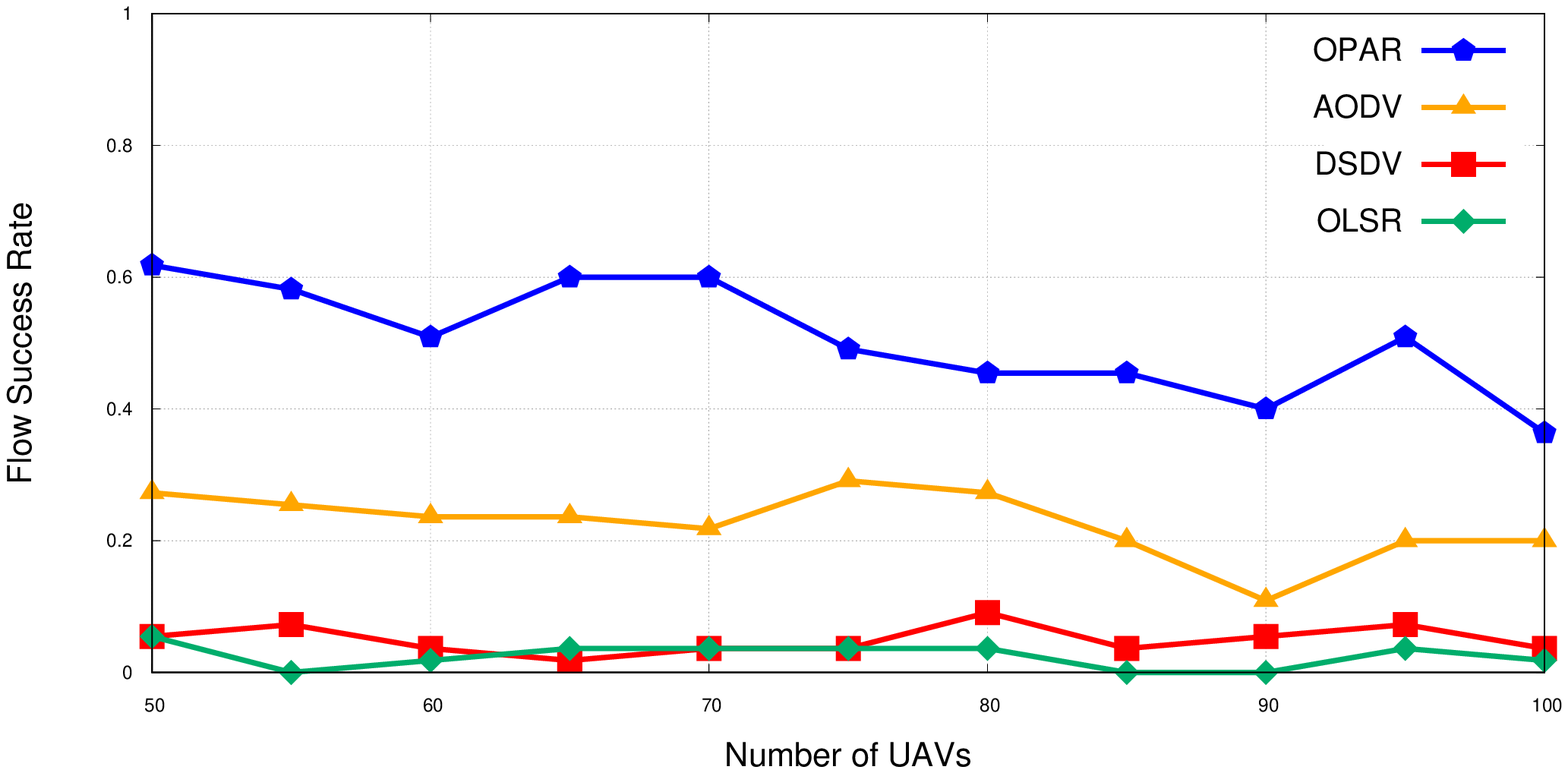}}
	\subfloat[Different number of UAVs (G-M)]{  \includegraphics[width=.5\linewidth]{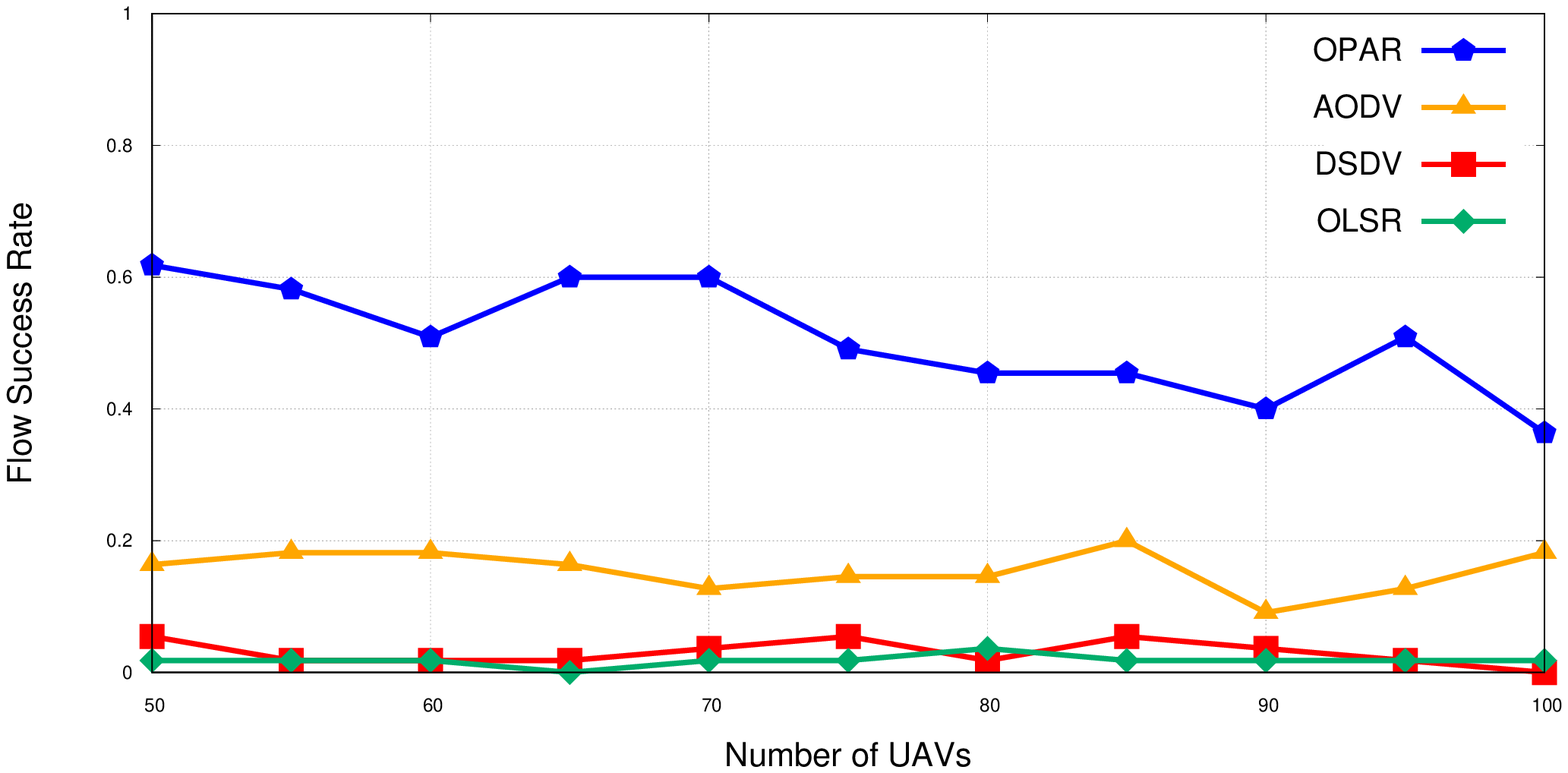}}
	\caption{A comparison of the flow success rate for different loads and densities.}
	\label{fig::success}
\end{figure} 

Next, we investigate the effect of varying the path length-weight on the FCT, in a network with different loads and different densities. We measure this metric in seconds and calculate it as the average FCT of flows in each simulation scenario. We consider the simulation time, i.e., 500 sec, as the lower bound for the flows which failed to deliver their files in the simulation time. Fig. (\ref{fig::wFCT}) shows the results. In accordance with the results of Fig. (\ref{fig::wthroughput}), for the network with each specific load and specific density, a certain value of the path length-weight leads to the highest performance. By comparing Fig. (\ref{fig::wFCT}) with Fig. (\ref{fig::wthroughput}), we find that, in most cases, the same values for the path length-weight that lead to higher throughput, also lead to lower FCT and consequently better overall performance. While not shown here, the OPAR success rate is not affected seriously by increasing the load and the density of the network, which shows its stability and reliability. 

Now, we compare the performance of OPAR with that of AODV, OLSR, and DSDV for both RWP and G-M mobility models. Fig. (\ref{fig::success}) compares the results of the flow success rate. We note that the overhead of conventional routing algorithms prevents them from successfully delivering their files, in highly dynamic networks. As Fig. (\ref{fig::success}) shows, the flow success rate of OPAR is on average $30\%$ higher than that of AODV and much higher than other routing protocols. This figure shows that the flow success rate decreases slightly by increasing the network load and network density. It also shows that the flow success rate in the scenarios under the RWP mobility model is higher than that of G-M. The main reason is that the randomness of the RWP model increases the chance of connectivity. 

Fig. (\ref{fig::routing}) shows the routing traffic in MB generated by different routing algorithms. For OPAR, since the routing operation is performed by the ground station, and the communication of UAVs with the ground station is on a channel other than the UAVs communicating channel, it might not be fair to compare its routing overhead with that of conventional routing algorithms. However, it is worth noting that the average OPAR routing overhead is only $0.83$ MB for a network with 50 UAVs which is less than $10\%$ of AODV routing overhead. This figure shows that OLSR has the highest routing traffic overhead, where AODV shows a better performance than DSDV. It shows that by increasing the network load, the routing traffic overhead increases slightly. However, increasing the number of nodes increases the routing traffic exponentially, even for a moderate network load. In this figure, we see how much the highly dynamic nature of the network plays a role in the dramatic increment of routing overhead, especially for OLSR. Furthermore, the routing traffic generated by conventional routing algorithms in the scenarios under the RWP mobility model is significantly less than that of G-M, owing to the higher network connectivity under RWP mobility.

\begin{figure}[t!]
	\centering
	\subfloat[Different number of  flows (RWP)]{\includegraphics[width=.5\linewidth]{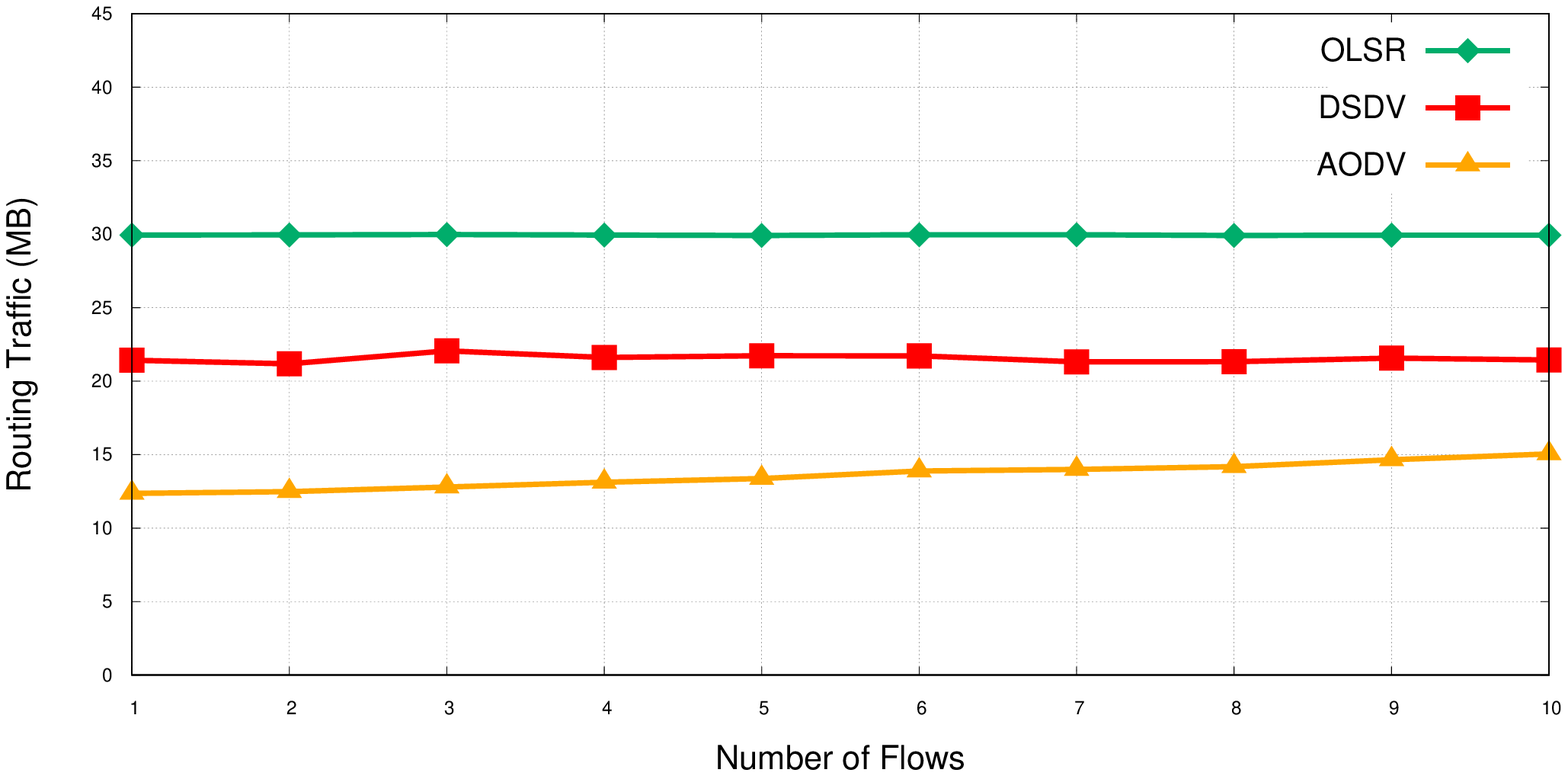}}
	\subfloat[Different number of  flows (G-M)]{\includegraphics[width=.5\linewidth]{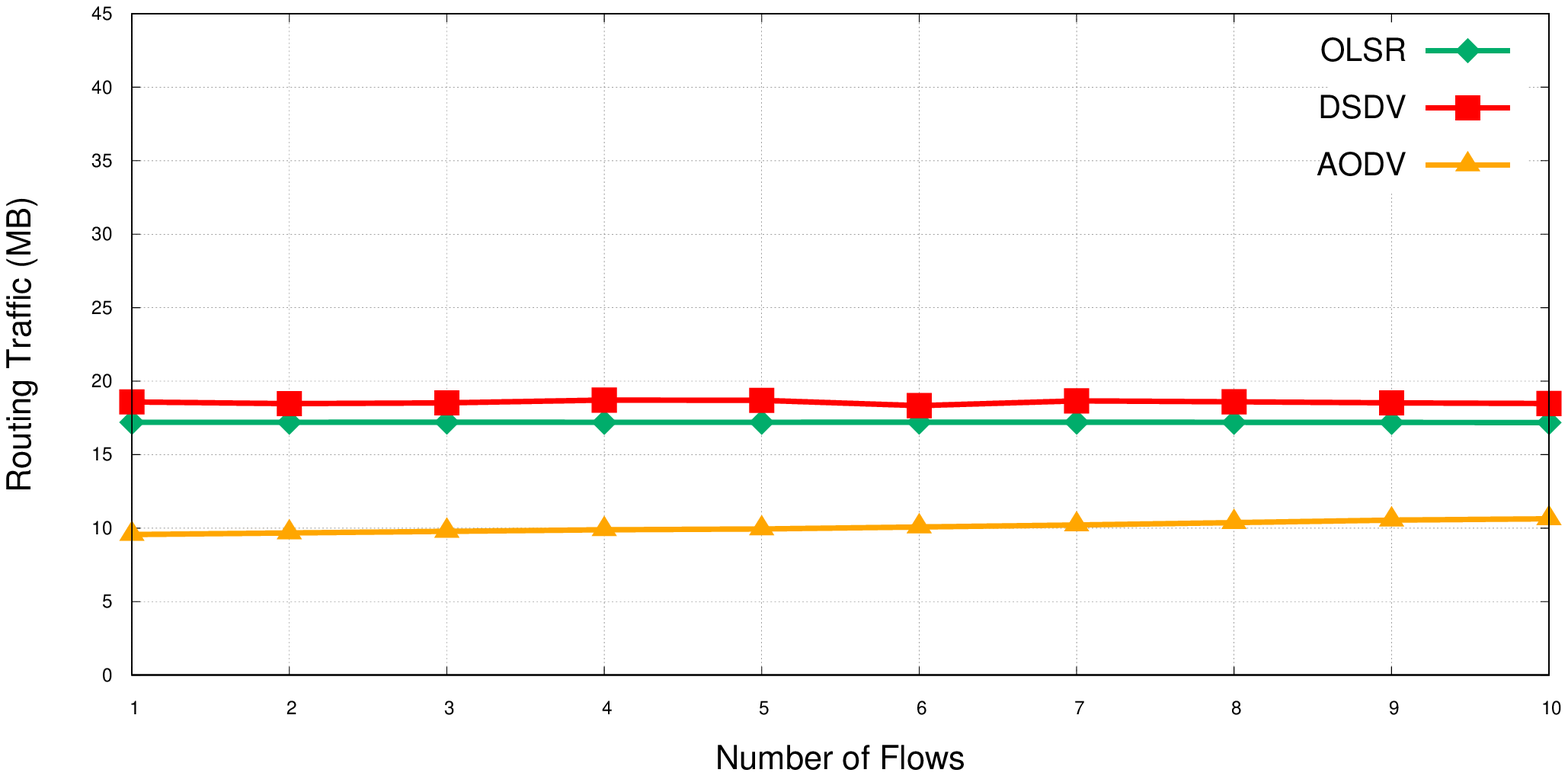}}\\
	\subfloat[Different number of UAVs (RWP)]{  \includegraphics[width=.5\linewidth]{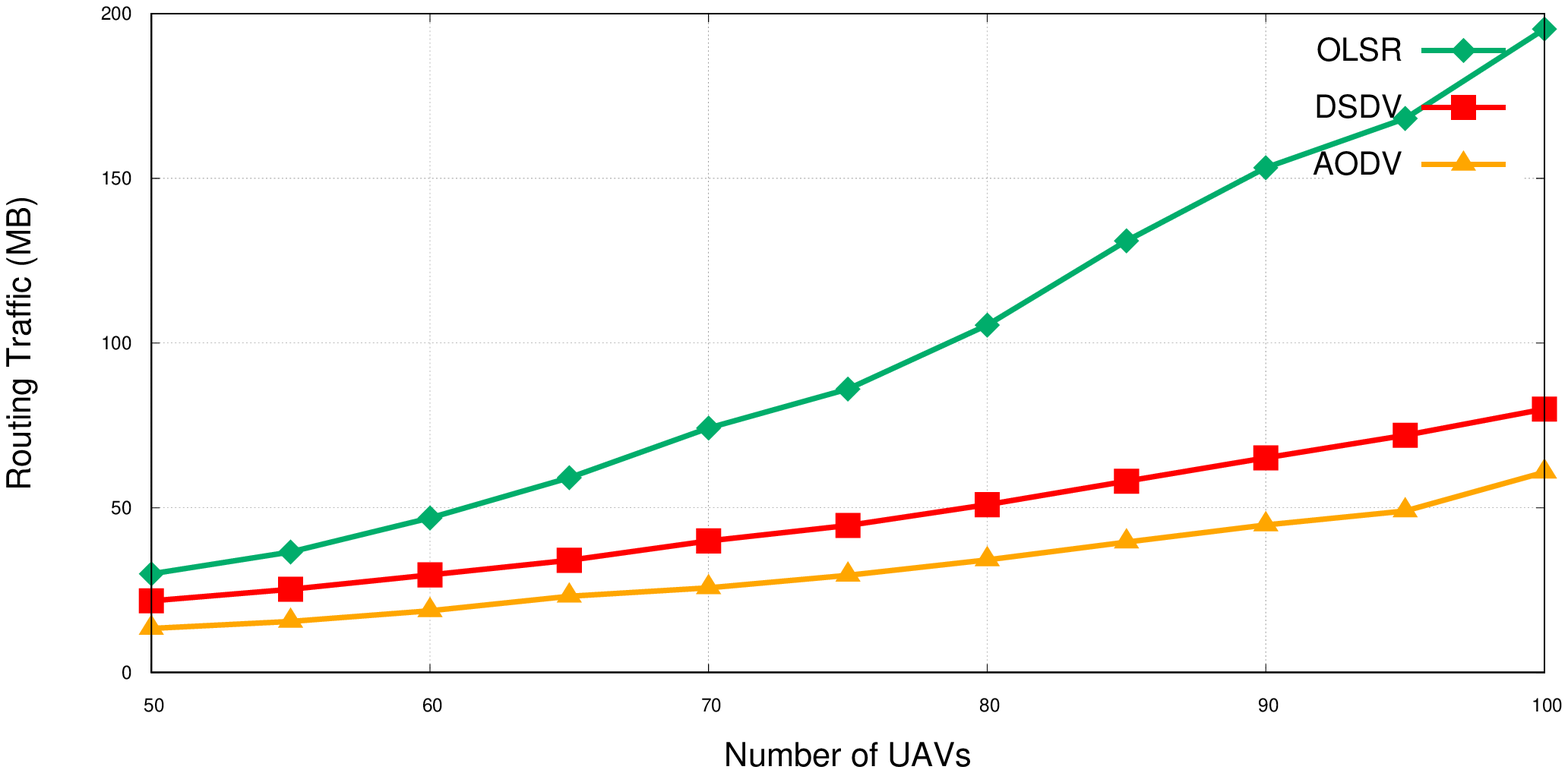}}
	\subfloat[Different number of UAVs (G-M)]{  \includegraphics[width=.5\linewidth]{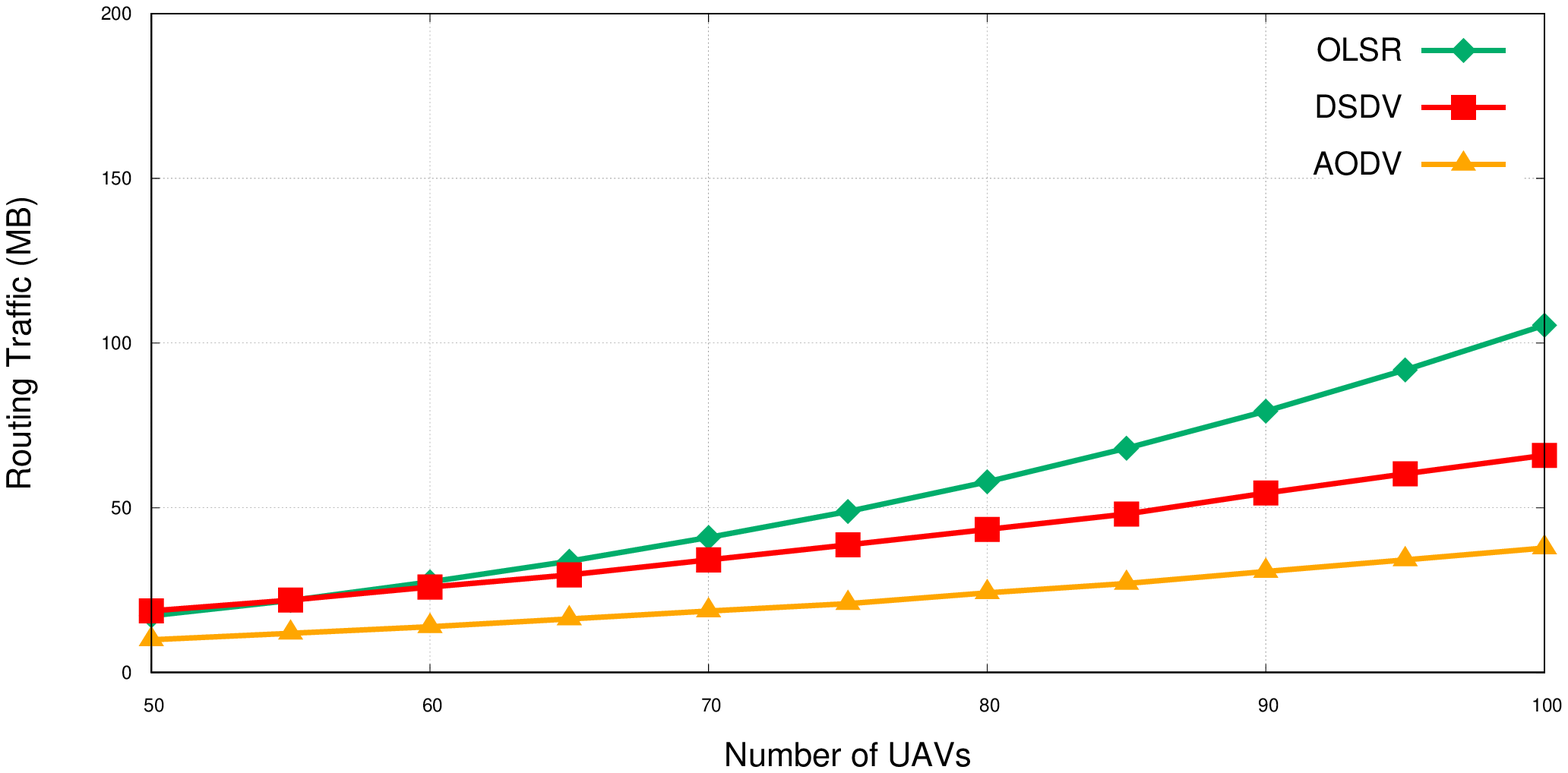}}
	\caption{A comparison of routing overhead for different loads and densities.}
	\label{fig::routing}
\end{figure}

Fig. (\ref{fig::throughput}) shows the comparison of network throughput for different network loads and densities. As it is shown in Fig. (\ref{fig::throughput}), on average, OPAR gains $25\%$ higher throughput in comparison with AODV, which outperforms DSDV and OLSR. This Figure shows that the increment in network load affects the throughput negatively due to a higher congestion rate. It further shows that the increment in the number of nodes degrades the throughput of the network. We observed that this performance degradation has two main reasons. First, by increasing the density of the network, the routing traffic of conventional routing algorithms increases dramatically. Second, by increasing the network density some paths with longer lengths are appeared and considered for data transfer. The longer paths generally increase congestion and lead to performance degradation. Comparing Fig. (\ref{fig::throughput}a) with Fig. (\ref{fig::throughput}b) and also Fig. (\ref{fig::throughput}c) with Fig. (\ref{fig::throughput}d), we can conclude that using longer paths has more of a negative effect than that of routing overhead. In the presence of the G-M model, the routing overhead is less than that of RWP which is supposed to improve the throughput. However, in the same case, some longer paths are used due to the lower connectivity in the network which degrades the throughput. We see that the throughput of the network under the G-M mobility model is lower than that of RWP. However, OPAR shows stable throughput for both G-M and RWP models.            

\begin{figure}[t!]
	\centering
	\subfloat[Different number of  flows (RWP)]{\includegraphics[width=.5\linewidth]{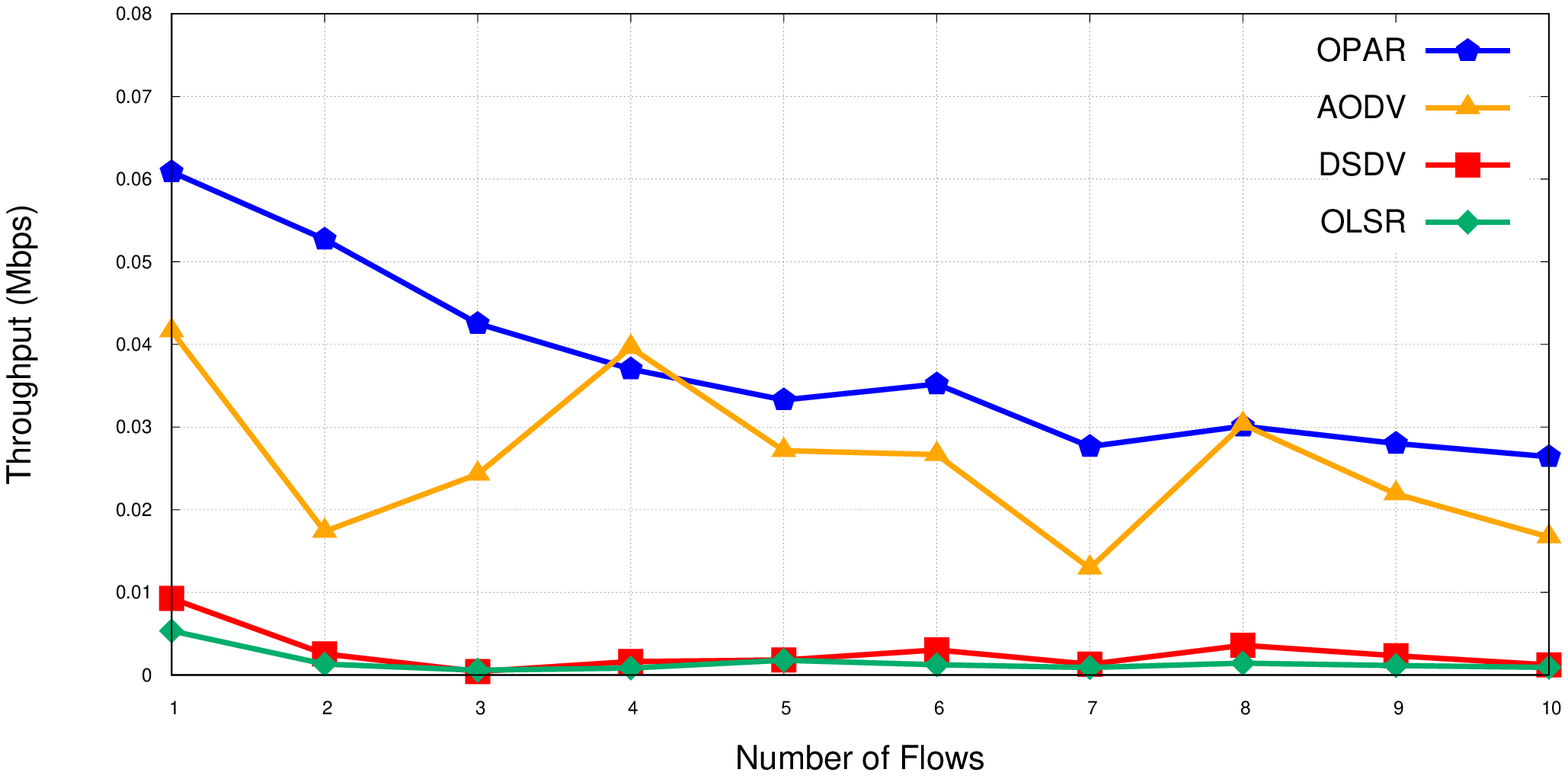}}
	\subfloat[Different number of flows (G-M)]{\includegraphics[width=.5\linewidth]{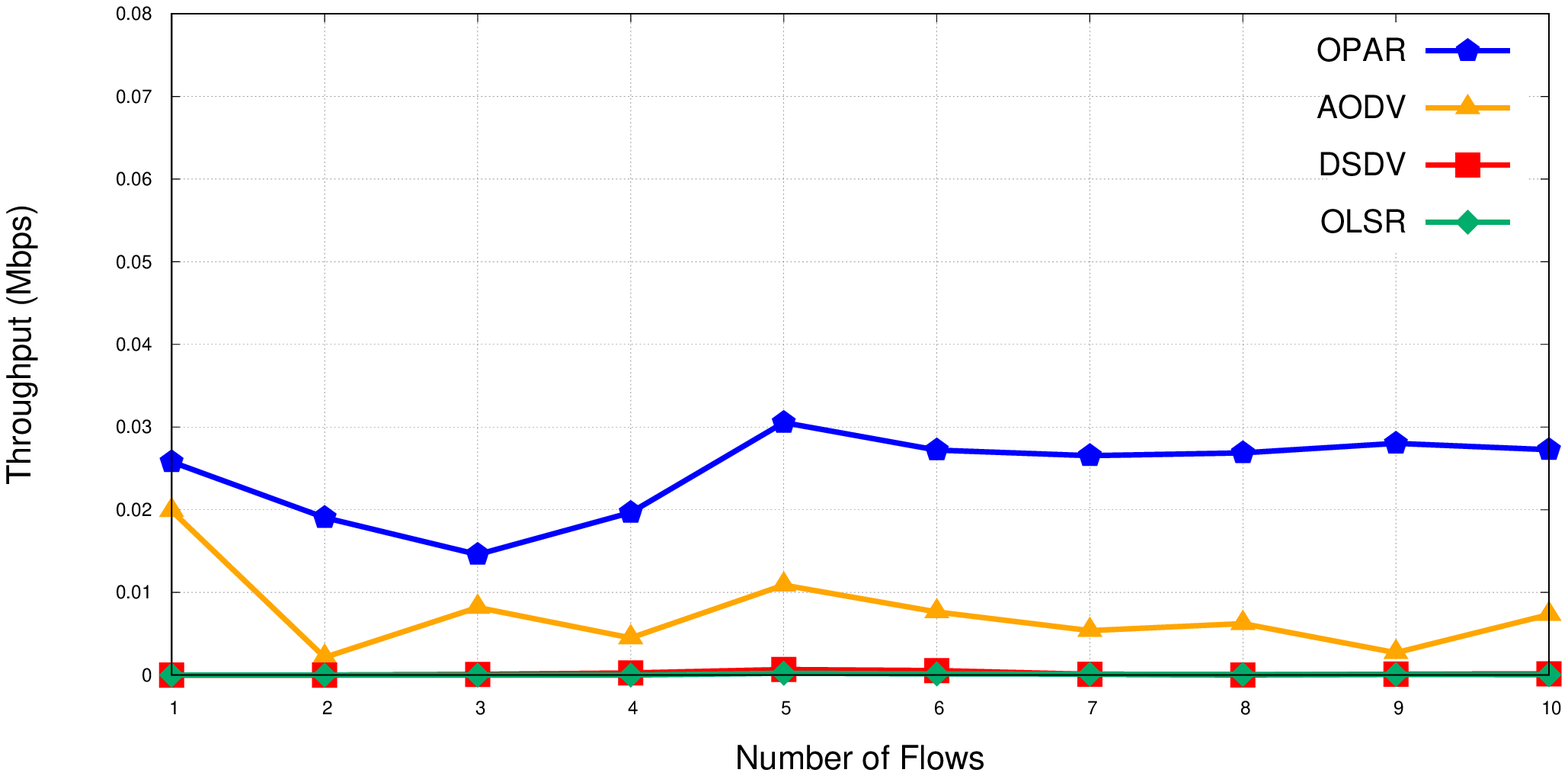}}\\
	\subfloat[Different number of UAVs (RWP)]{  \includegraphics[width=.5\linewidth]{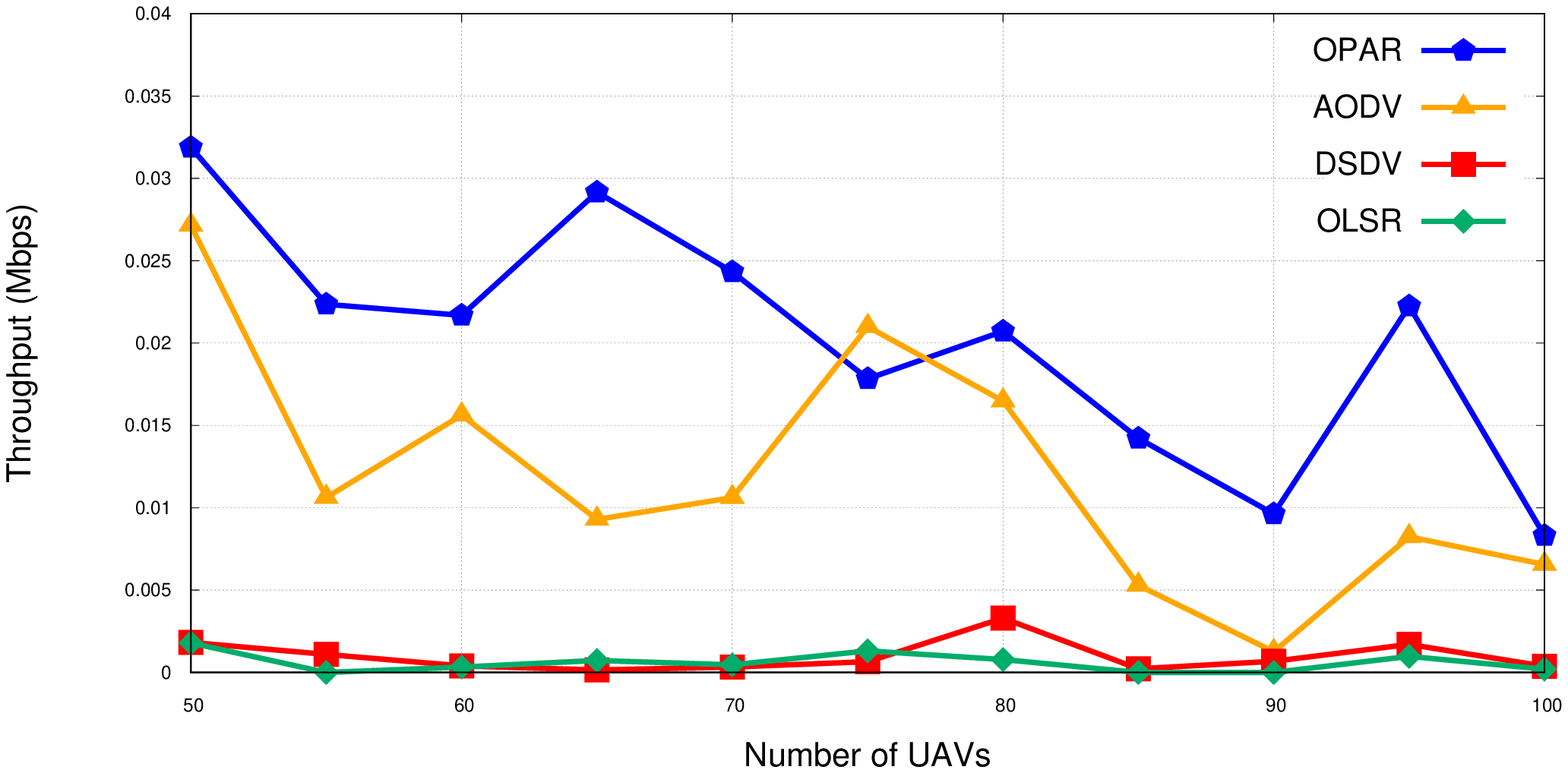}}
	\subfloat[Different number of UAVs (G-M)]{  \includegraphics[width=.5\linewidth]{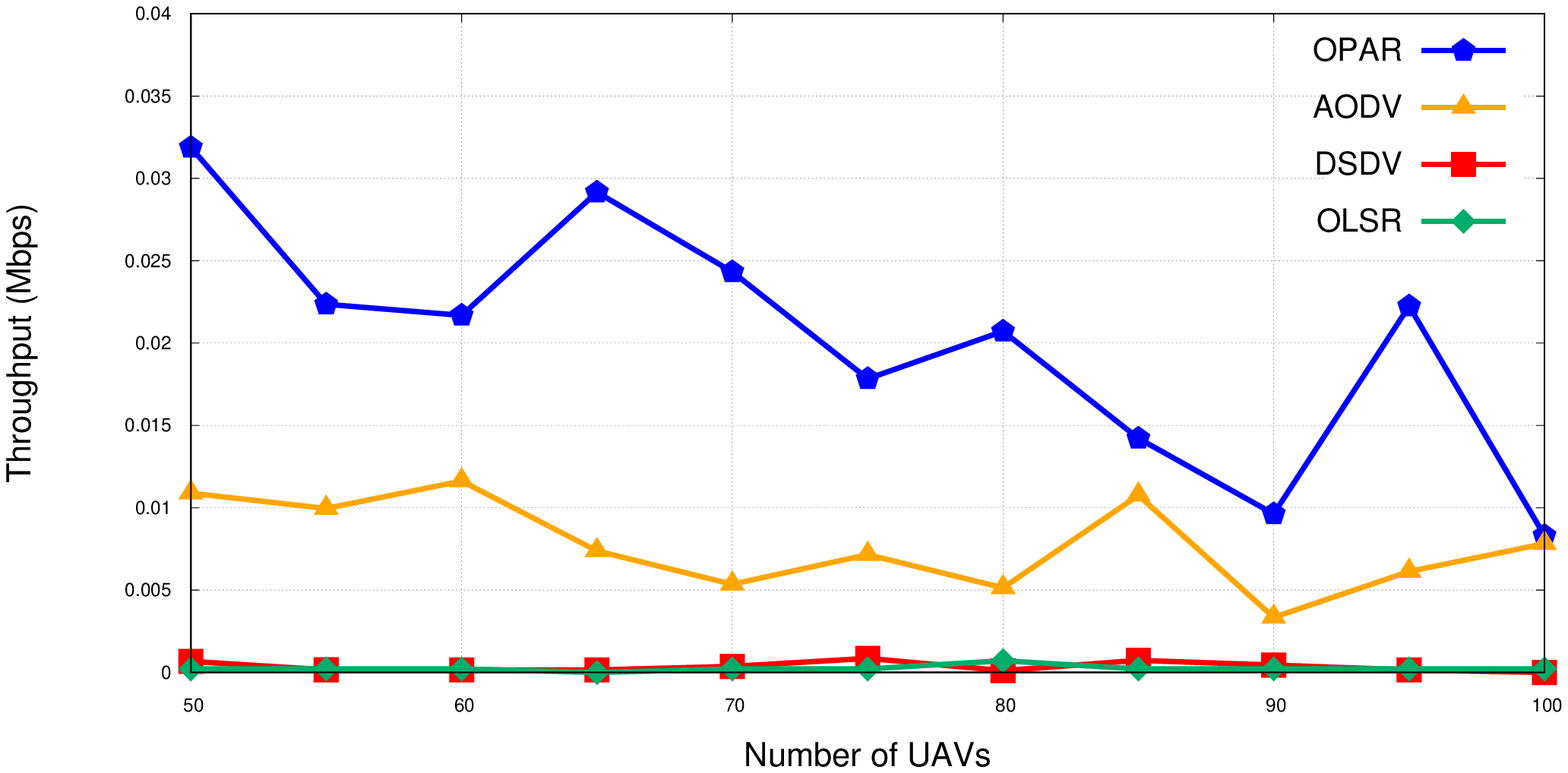}}
	\caption{A comparison of the throughout for different loads and densities.}
	\label{fig::throughput}
\end{figure}

Finally, we show the results of FCT in Fig. (\ref{fig::FCT}). This figure shows that OPAR needs $15\%$ less time, on average, in comparison with AODV, and around $50\%$ less in comparison with DSDV and OLSR. While the increase in the network load does not show a significant increase in FCT, increasing the network density significantly increases the FCT. The negative effect of using the G-M model is also obvious where the number of concurrent flows is increased which arises from the increment in the number of failed flows in the G-M model. 

\begin{figure}[t!]
	\centering
	\subfloat[Different number of  flows (RWP)]{\includegraphics[width=.5\linewidth]{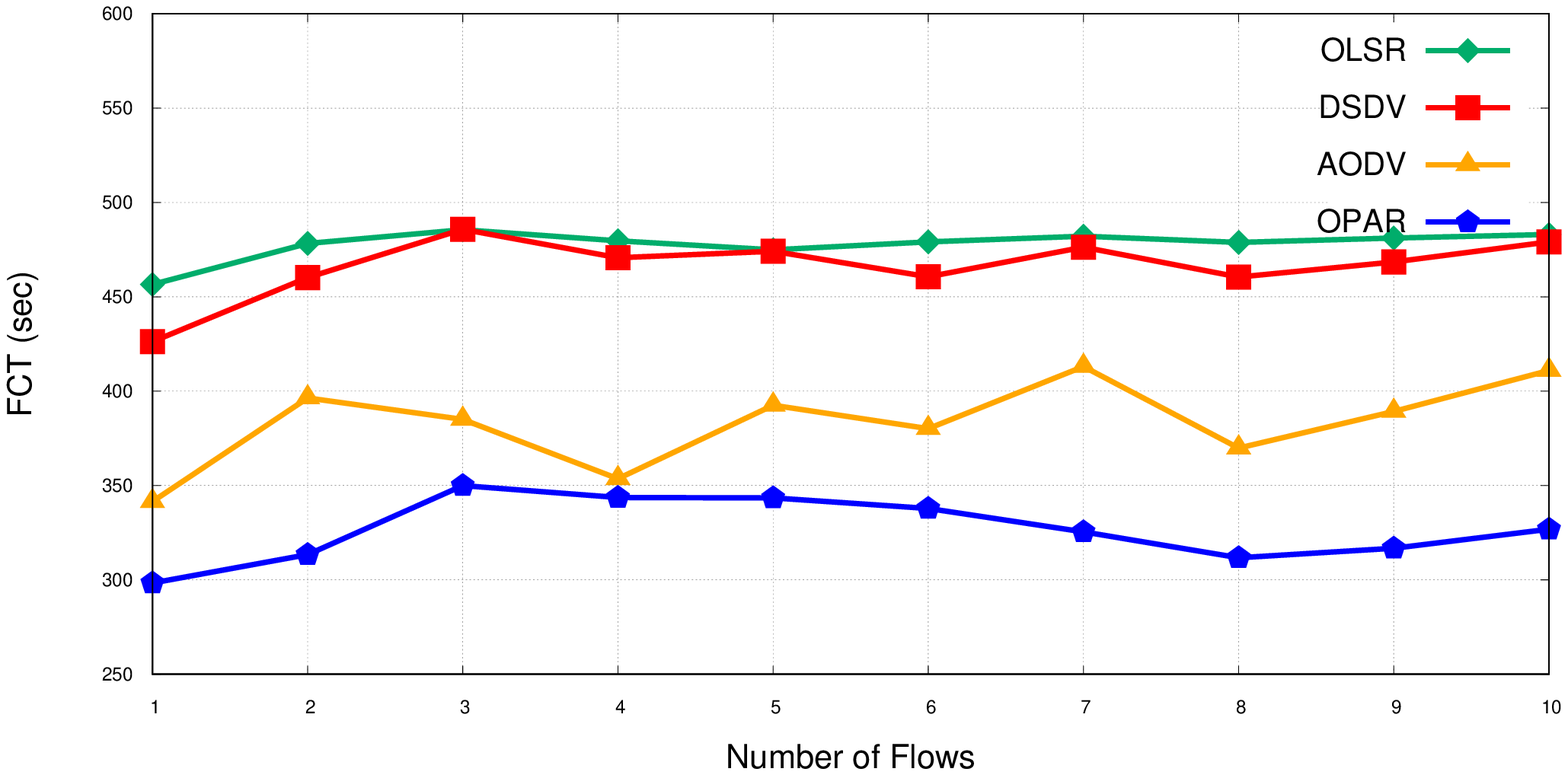}}
	\subfloat[Different number of flows (G-M)]{\includegraphics[width=.5\linewidth]{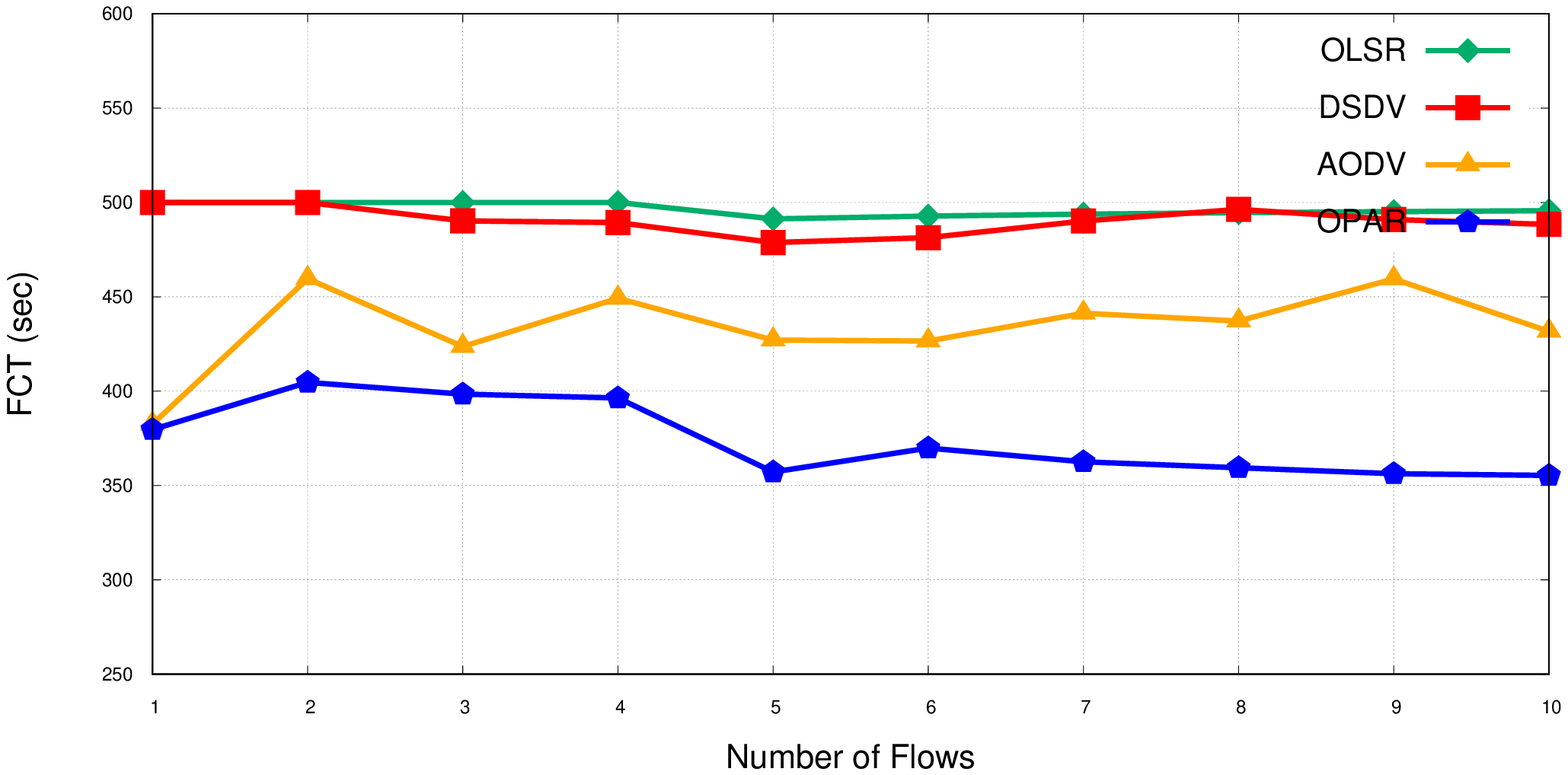}}\\
	\subfloat[Different number of UAVs (RWP)]{  \includegraphics[width=.5\linewidth]{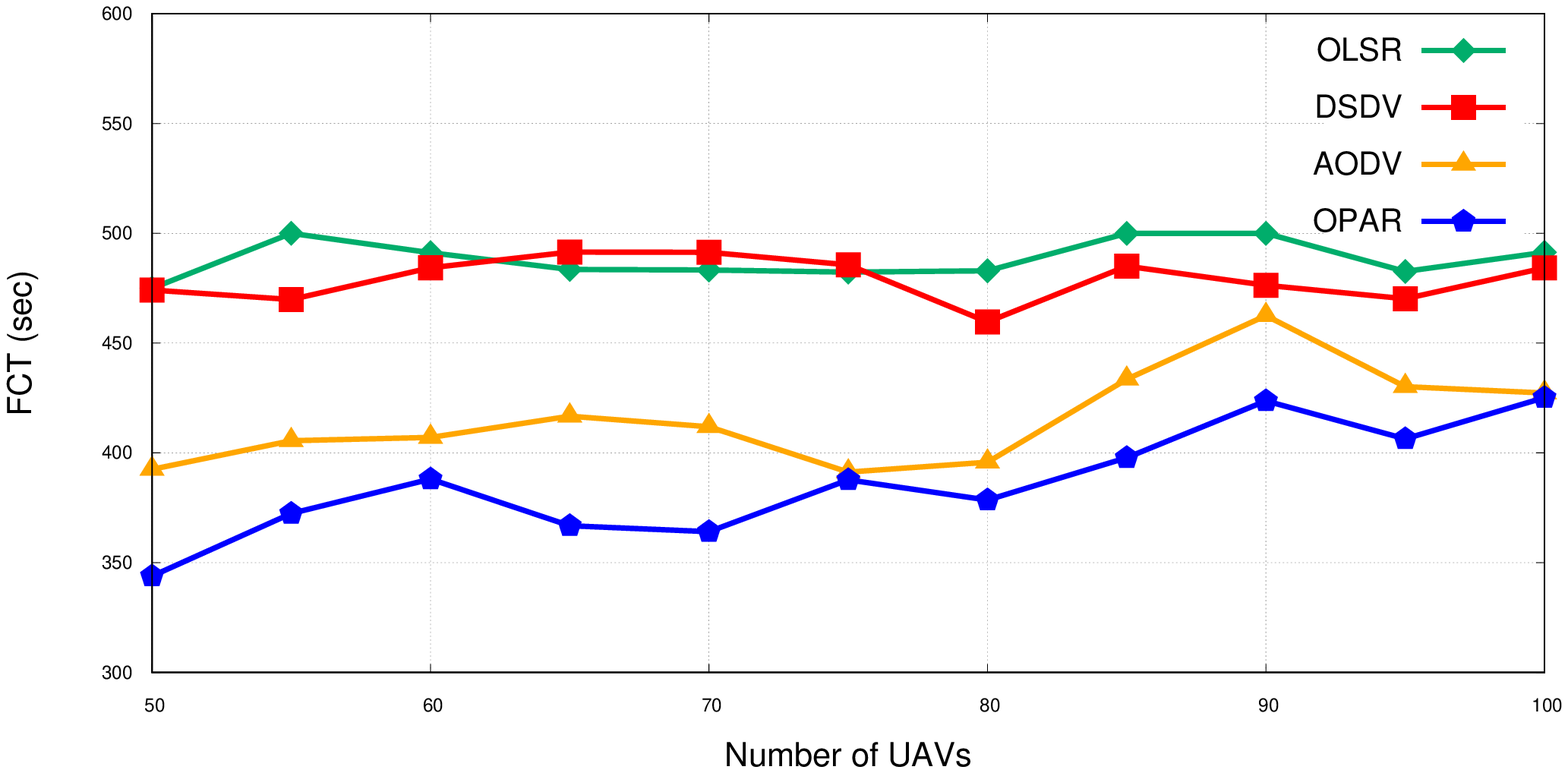}}
	\subfloat[Different number of UAVs (G-M)]{  \includegraphics[width=.5\linewidth]{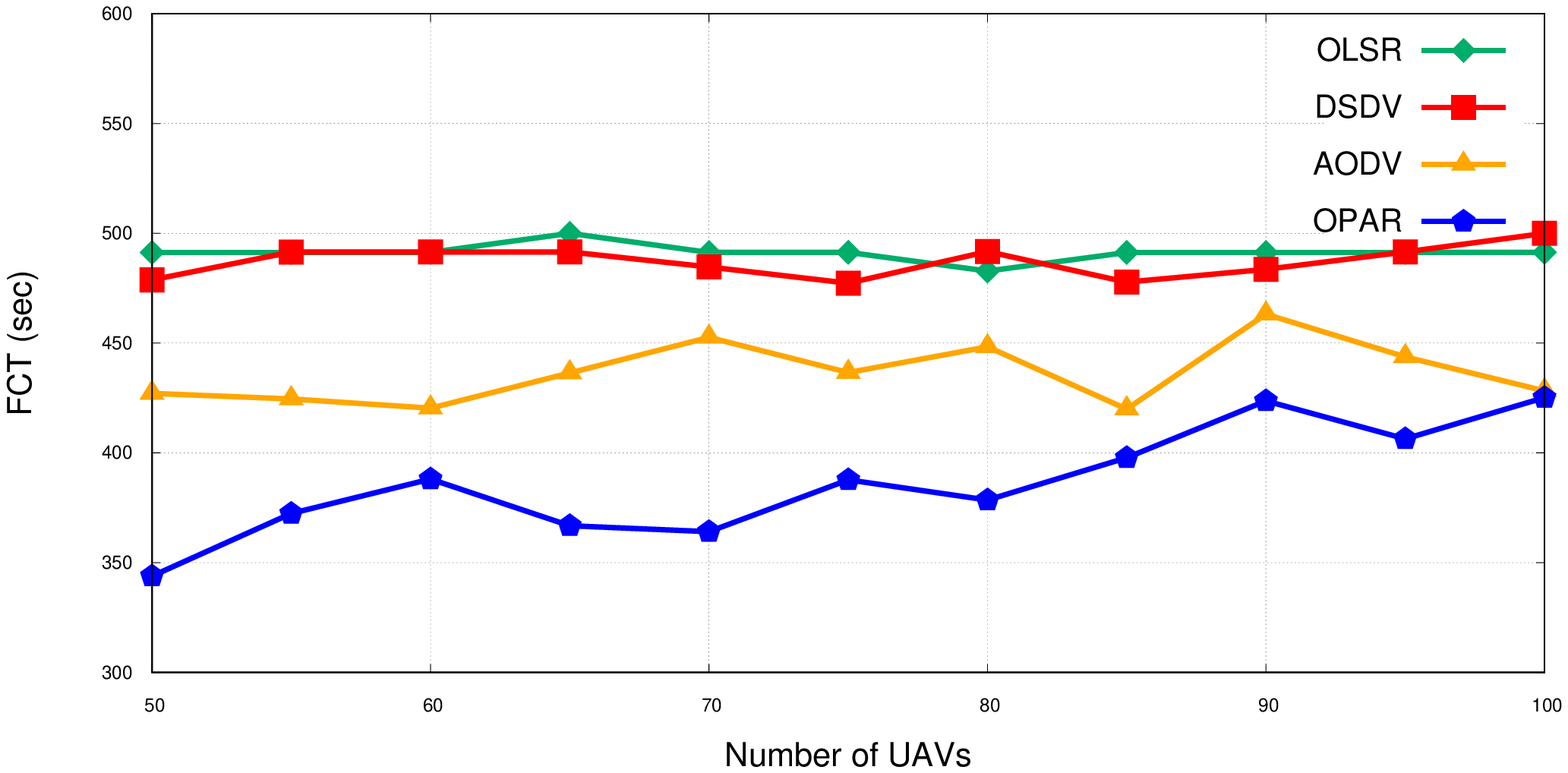}}
	\caption{A comparison of FCT for different loads and densities.}
	\label{fig::FCT}
\end{figure}

%% file: conclusion.tex
\section{Conclusion}
\label{sec::conclusion}

The high dynamicity of FANETs, along with their other specific characteristics, leads to the incompetence of the conventional routing algorithms. While we showed that the routing traffic overhead of such algorithms could be even higher than the transferred data in highly dynamic networks, they are very likely to fail in delivering the flow of data. To face this shortcoming, we proposed OPAR, an optimized predictive and adaptive routing solution that optimizes the network performance by considering the path lifetime and the path length. We exhaustively evaluated the network's performance using RWP and G-M mobility models for different loads and densities. All the combinations show the superiority of OPAR for the measured flow success rate, throughput, and flow completion time. As the future directions, we aim at adding the load balancing module to the proposed OPAR and develop the fully distributed version of the algorithm.

\section{acknowledgment}
This material is in part based upon work funded by AFRL Grant \#FA8750-20-1-1000. Any opinions, findings and conclusions or recommendations expressed in this material are those of the author(s) and do not necessarily reflect the views of the US government or AFRL.